\def\year{2020}\relax
\documentclass[letterpaper]{article} 
\usepackage{aaai20}  
\usepackage{times}  
\usepackage{helvet} 
\usepackage{courier}  
\usepackage[hyphens]{url}  
\usepackage{graphicx} 
\usepackage{graphicx}  
\usepackage{amsmath,amssymb,amsfonts,bm,amsthm}
\usepackage{autobreak}
\usepackage{algorithmic,algorithm}
\usepackage{array}
\usepackage{textcomp}
\usepackage{xcolor}
\usepackage{epstopdf}
 \usepackage{multirow}
\usepackage{mathrsfs}
\usepackage{subfigure}
\newtheorem{theorem}{Theorem}
\newtheorem{lemma}{Lemma}

\newtheorem{definition}{Definition}

\newcommand{\x}{\bm{x}}

\title{Differentially Private and Fair Classification via Calibrated Functional Mechanism}

\author{ Jiahao Ding\textsuperscript{\rm 1}, Xinyue Zhang\textsuperscript{\rm 1}, Xiaohuan Li\textsuperscript{\rm 2}, Junyi Wang\textsuperscript{\rm 2}, Rong Yu\textsuperscript{\rm 3}, Miao Pan\textsuperscript{\rm 1}\\
\textsuperscript{\rm 1}{University of Houston}\\ 
\textsuperscript{\rm 2}{Guilin University of Electronic Technology}\\
\textsuperscript{\rm 3}{Guangdong University of Technology}\\
\{jding7, xzhang67, mpan2\}@uh.edu, \{lxhguet, wangjy\}@guet.edu.cn, yurong@gdut.edu.cn
}

\begin{document}

\maketitle

\begin{abstract}
Machine learning is increasingly becoming a powerful tool to make decisions in a wide variety of applications, such as medical diagnosis and autonomous driving. Privacy concerns related to the training data and unfair behaviors of some decisions with regard to certain attributes (e.g., sex, race) are becoming more critical. Thus, constructing a fair machine learning model while simultaneously providing privacy protection becomes a challenging problem. 
In this paper, we focus on the design of classification model with fairness and differential privacy guarantees by jointly combining functional mechanism and decision boundary fairness. In order to enforce $\epsilon$-differential privacy and fairness, we leverage the functional mechanism to add different amounts of Laplace noise regarding different attributes to the polynomial coefficients of the objective function in consideration of fairness constraint.
We further propose an utility-enhancement scheme, called relaxed functional mechanism by adding Gaussian noise instead of Laplace noise, hence achieving $(\epsilon,\delta)$-differential privacy. Based on the relaxed functional mechanism, we can design $(\epsilon,\delta)$-differentially private and fair classification model.
Moreover, our theoretical analysis and empirical results demonstrate that our two approaches achieve both fairness and differential privacy while preserving good utility and outperform the state-of-the-art algorithms. 
\end{abstract}

\section{Introduction}
In this big data era, machine learning has been becoming a powerful technique for automated and data-driven decision making processes in various domains, such as spam filtering, credit ratings, housing allocation, and so on. However, as the success of machine learning mainly rely on a vast amount of individual data (e.g., financial transactions, tax payments), there are growing concerns about the potential for privacy leakage and unfairness in training and deploying machine learning algorithms~\cite{fredrikson2015model,datta2015automated}. Thus, the problem of fairness and privacy in machine learning has attracted considerable attention.

Fairness-aware learning has received growing attentions in the machine learning field due to the social inequities and unfair behaviors observed in classification models. For example, a classification model of automated job hiring system is more likely to hire candidates from certain racial or gender groups \cite{Vivian,Sara}. Hence, substantial effort has centered on developing algorithmic methods for designing fair classification models and balancing the trade-off between accuracy and fairness, mainly including two groups: pre/post-processing methods \cite{dwork2012fairness,feldman2015certifying,hardt2016equality} and in-processing methods \cite{kamishima2011fairness,zafar2015fairness}. Pre/post-processing methods achieve fairness by directly changing values of the sensitive attributes or class labels in the training data. As pointed out in \cite{zafar2015fairness}, pre/post-processing methods treat the learning algorithm as a black box, which can result in unpredictable loss of the classification utility. Thus, in-processing methods, which introduce fairness constraints or regularization terms to the objective function to remove the discriminatory effect of classifiers, have been shown a great success.

At the same time, differential privacy \cite{dwork2014algorithmic} has emerged as the de facto standard for measuring the privacy leakage associated with algorithms on sensitive databases, which has recently received considerable attentions by large-scale corporations such as Google \cite{erlingsson2014rappor} and Microsoft \cite{ding2017collecting}, etc. Generally speaking, differential privacy ensures that there is no statistical difference to the output of a randomized algorithm whether a single individual opts in to, or out of its input. A large class of mechanisms has been proposed to ensure differential privacy. For instance, the Laplace mechanism is employed by introducing random noise drawn from the Laplace distribution to the output of queries such that the adversary will not be able to confirm a single individual is in the input with high confidence \cite{dwork2006calibrating}. To design private machine learning models, more complicated perturbation mechanisms have been proposed like objective perturbation \cite{chaudhuri2011differentially} and functional mechanism \cite{zhang2012functional}, which inject random noise into the objective function rather than model parameters.

Thus, in this paper, we mainly focus on achieving classification models that simultaneously provide differential privacy and fairness. As pointed out in recent study \cite{xu2019achieving}, achieving both requirements efficiently is quite challenging, due to the different aims of differential privacy and fairness. Differential privacy in a classification model focuses on the individual level, i.e., differential privacy guarantees that the model output is independent of whether any individual record presents or absents in the dataset, while fairness in a classification model focuses on the group level, i.e., fairness guarantees that the model predictions of the protected group (such as female group) are same to those of the unprotected group (such as male group). Lots of researches have emerged in achieving both privacy protection and fairness. Specifically, in \cite{dwork2012fairness}, Dwork et al. gave a new definition of fairness that is an extended definition of differential privacy. In \cite{hajian2015discrimination}, Hajian et al. imposed fairness and $k$-anonymity via a pattern sanitization method. Moreover, Ekstrand et al. in \cite{ekstrand2018privacy} put forward a set of questions about whether fairness are compatible with privacy. However, only Xu et al. in \cite{xu2019achieving} studied how to meet the requirements of both differential privacy and fairness in classification models by combining functional mechanism and decision boundary fairness together. Therefore, how to simultaneously meet the requirements of differential privacy and fairness in machine learning algorithms is under exploited. 

In this paper, we propose \textbf{P}urely and \textbf{A}pproximately \textbf{D}ifferential private and \textbf{F}air \textbf{C}lassification algorithms, called PDFC and ADFC, respectively, by incorporating functional mechanism and decision boundary covariance, a novel measure of decision boundary fairness.
As shown in \cite{kamiran2012data}, due to the correlation between input features (attributes), the discrimination of classification still exists even if removing the protected attribute from the dataset before training. Hence, different from \cite{xu2019achieving}, which adds same scale of noise in each attribute, in PDFC, we consider a calibrated functional mechanism, i.e., injecting different amounts of Laplace noise regarding different attributes to the polynomial coefficients of the constrained objective function to ensure $\epsilon$-differential privacy and reduce effects of discrimination. To further improve the model accuracy, in ADFC, we propose a relaxed functional mechanism by inserting Gaussian noise instead of Laplace noise and leverage it to perturb coefficients of the polynomial representation of the constrained objective function to enforce $(\epsilon,\delta)$-differential privacy and fairness. Our salient contributions are listed as follows.
\begin{itemize}
    \item We propose two approaches PDFC and ADFC to learn a logistic regression model with differential privacy and fairness guarantees by applying functional mechanism to a constrained objective function of logistic regression that decision boundary fairness constraint is treated as a penalty term and added to the original objective function.
    \item For PDFC, different magnitudes of Laplace noise regarding different attributes are added to the polynomial coefficients of the constrained objective function to enforce $\epsilon$-differential privacy and fairness.
    \item For ADFC, we further improve the model accuracy by proposing the relaxed functional mechanism based on Extended Gaussian mechanism, and leverage it to introduce Gaussian noise with different scales to perturb objective function.
    \item Using real-world datasets, we show that the performance of PDFC and ADFC significantly outperforms the baseline algorithms while jointly providing differential privacy and fairness.
\end{itemize}
The rest of paper is organized as follows. We first give the problem statement and background in differential privacy and fairness. Next, we present our two approaches PDFC and ADFC to achieve DP and fair classification. Finally, we give the numerical experiments based on real-world datasets and draw conclusion remarks. Due to the space limit, we leave all the proofs in the supplemental materials.

\section{Problem Statement}

This paper considers a training dataset $\mathscr{D}$ that includes $n$ tuples $t_1,t_2,\cdots,t_n$. We also denote each tuple $t_i=(\x_i,y_i)$ where the feature vector $\x_i$ contains $d$ attributes, i.e., $\x_i=(x_{i1},x_{i2},\cdots,x_{id})$, and $y_i$ is the corresponding label. Without loss of generality, we assume $\sqrt{\sum_{j=1}^d x_{ij}^2}\leq 1$ where $x_{ij}\geq 0$, and $y_i \in \{0,1\}$ for binary classification tasks. The objective is to construct a binary classification model $\rho(\x,w)$ with model parameters $w=(w_1,w_2,\cdots,w_d)$ that taken $\x$ as input, can output the prediction $\hat{y}$, by minimizing the empirical loss on the training dataset $\mathscr{D}$ over the parameter space $w$ of $\rho$.

In general, we have the following optimization problem.
\begin{align}
    w^*&=\arg \min_{w} f(\mathscr{D},w)=\arg \min_{w}\sum_{i=1}^n f(t_i,w)
\end{align}
where $f$ is the loss function. In this paper, we consider logistic regression as the loss function, i.e., $f(\mathscr{D},w)=\sum_{i=1}^n[\log(1+exp(\x_i^T w))-y_i\x_i^T w]$. Thus, the classification model has the form $\rho(\x,w^*)=\frac{exp(\x^T w^*)}{1+exp(\x^T w^*)}$.

Although there is no need to share the dataset during the training procedure, the risk of information leakage still exists when we release the classification model parameter $w^*$. For example, the adversary may perform model inversion attack \cite{fredrikson2015model} over the release model $w^*$ together with some background knowledge about the training dataset to infer sensitive information in the dataset. 

Furthermore, if labels in the training dataset are associated with a protected attribute $z_i$ (note that we denote $\x_i$ as unprotected attributes), like gender, the classifier may be biased, i.e., $P(\hat{y_i}=1|z_i=0) \ne P(\hat{y_i}=1|z_i=1)$, where we assume the protected attribute $z_i\in \{0,1\}$. According to \cite{pedreshi2008discrimination}, even if the protected attribute is not used to build the classification model, this unfair behavior may happen when the protected attribute is correlated with other unprotected attributes.

Therefore, in this paper, our objective is to learn a binary classification model, which is able to guarantee differential privacy and fairness while preserving good model utility. 

\section{Background}
In this section, we first introduce some background knowledge of differential privacy, which helps us to build private classification models. Then we present fairness definition, which helps us to enforce classification fairness.

\subsection{Differential Privacy}
Differential privacy is introduced to guarantee that the ability of an adversary to obtain additional information about any individual is independent of whether any individual record presents or absents in the dataset.
\begin{definition}[\bm{$\epsilon$}\bf{-Differential Privacy}]
A randomized Mechanism $\mathcal{A}$ is enforced by $\epsilon$-differential privacy, if for any two neighboring datasets $\mathscr{D},\mathscr{D}' \in \mathbb{D}$, i.e., differing at most one single data sample, and for any possible output $s$ in the output space of $\mathcal{A}$, it holds that $\Pr(\mathcal{A}(\mathscr{D}) = s )\leq e^{\epsilon} \Pr(\mathcal{A}(\mathscr{D}') = s).$
\end{definition}
The privacy parameter $\epsilon$ controls the strength of the privacy guarantee. A smaller value indicates a stronger privacy protection. 
Though differential privacy provides very strong guarantee, in some cases it may be too strong to have a good data utility. We then introduce a relaxation, ($\epsilon$,$\delta$)-differential privacy, that has been proposed in \cite{dwork2006our}.
\begin{definition}[(\bm{$\epsilon$},\bm{$\delta$})\bf{-Differential Privacy}]
A randomized Mechanism $\mathcal{A}$ is enforced by ($\epsilon$,$\delta$)-differential privacy, if for any two neighboring datasets $\mathscr{D},\mathscr{D}' \in \mathbb{D}$ differing at most one single data item, and for any possible output $s$ in the output space of $\mathcal{A}$, it holds that $Pr(\mathcal{A}(\mathscr{D}) = s )\leq e^{\epsilon} Pr(\mathcal{A}(\mathscr{D}') = s)+\delta$.
\end{definition}
Laplace mechanism~\cite{dwork2014algorithmic} and Extended Gaussian mechanism \cite{DBLP:journals/corr/abs-1906-01444} are common techniques for achieving differential privacy, both of which add random noise calibrated to the sensitivity of the query function $q$. 
\begin{theorem}[\bf{Laplace Mechanism}]
Given any function $q: \mathcal{X}^n \rightarrow \mathbb{R}^d$, the Laplace mechanism defined by 
\begin{align*}
    \mathcal{M}_L(\mathscr{D},q,\epsilon) = q(\mathscr{D}) + (\mathscr{Y}_1, \mathscr{Y}_2, \cdots, \mathscr{Y}_d)
\end{align*} 
preserves $\epsilon$-differential privacy, where $\mathscr{Y}_i$ are i.i.d. random variables drawn from $Lap(\Delta_{1}q/\epsilon)$ and $l_1$-sensitivity of the query $q$ is $ \Delta_{1}q = \sup_{\substack{\mathscr{D},\mathscr{D}'}}~\|q(\mathscr{D})-q(\mathscr{D}')\|_1$ taken over all neighboring datasets $\mathscr{D}$ and $\mathscr{D}'$.
\end{theorem}
\begin{theorem}[\bf{Extended Gaussian Mechanism}]\label{gaussin}
Given any function $q: \mathcal{X}^n \rightarrow \mathbb{R}^d$ and for any $\epsilon>0$, $\delta \in (0,1)$, the Extended Gaussian mechanism defined by \begin{align*}
    \mathcal{M}_G(\mathscr{D},q,\epsilon) = q(\mathscr{D}) +  (\mathscr{Y}_1, \mathscr{Y}_2, \cdots, \mathscr{Y}_d)
\end{align*} 
preserves $(\epsilon,\delta)$-differential privacy,
where $\mathscr{Y}_i$ are i.i.d drawn from a Gaussian distribution $\mathcal{N}(0, \sigma^2 I_d)$ with $\sigma\geq\frac{\sqrt{2}\Delta_{2q}}{2\epsilon}(\sqrt{\log (\sqrt{\frac{2}{\pi}}\frac{1}{\delta})}+\sqrt{\log (\sqrt{\frac{2}{\pi}}\frac{1}{\delta})+\epsilon})$ and $l_2$-sensitivity of the query $q$ is
$\Delta_{2}q = \sup_{\substack{\mathscr{D},\mathscr{D}'}}~\|q(\mathscr{D})-q(\mathscr{D}')\|_2$ taken over all neighboring datasets $\mathscr{D}$ and $\mathscr{D}'$.
\end{theorem} 

{\bf{Functional Mechanism.}} Functional mechanism, introduced by \cite{zhang2012functional}, as an extension of the Laplace mechanism is designed for regression analysis. To preserve $\epsilon$-differential privacy, functional mechanism injects differentially private noise into the objective function $f(\mathscr{D},w)$ and then publishs a noisy model parameter $\hat{w}$ derived from minimizing the perturbed objective function $\hat{f}(\mathscr{D},w)$ rather than the original one. As a result of the objective function being a complex function of $w$, in functional mechanism, $f(\mathscr{D},w)$ is represented in polynomial forms trough Taylor Expansion. The model parameter $w$ is a vector consisting of several values $w_1,w_2,\cdots,w_d$. We denote $\phi(w)$ as a product of $w_1,w_2,\cdots,w_d$, namely, $\phi(w)=w_1^{c_1} w_2^{c_2}\cdots w_d^{c_d}$ for some $c_1,c_2,\cdots,c_d \in \mathbb{N}$. We also denote $\Phi_j(j\in \mathbb{N})$ as the set of all products of $w_1,w_2,\cdots,w_d$ with degree $j$, i.e., $\Phi_j=\{w_1^{c_1} w_2^{c_2}\cdots w_d^{c_d}|\sum_{l=1}^d c_l = j\}$.

According to the Stone-Weierstrass Theorem \cite{rudin1964principles}, any continuous and differentiable function can always be expressed as a polynomial form. Therefore, the objective function $f(\mathscr{D},w)$ can be written as follows
\begin{align}\label{ployre}
    f(\mathscr{D},w)=\sum_{i=1}^n \sum_{j=0}^J \sum_{\phi \in \Phi_j}\lambda_{\phi t_{i}} \phi(w),
\end{align}
where $\lambda_{\phi t_{i}}$ represents the coefficient of $\phi(w)$ in polynomial.

To preserve $\epsilon$-differential privacy, the objective function $f(\mathscr{D},w)$ is perturbed by adding Laplace noise into the polynomial coefficients, i.e., ${\lambda}_{\phi}=\sum_{i=1}^n \lambda_{\phi t_{i}}+ Lap(\Delta_{1}/\epsilon)$, where $\Delta_1 = 2 \max_{t}\sum_{j=1}^J\sum_{\phi \in \Phi_j}\|\lambda_{\phi t}\|_1$. And then the model parameter $\hat{w}$ is obtained by minimizing the noisy objective function $\hat{f}(\mathscr{D},w)$. The sensitivity of logistic regression is given in the following lemma
\begin{lemma}[\bf{$l_1$-Sensitivity of Logistic Regression}]\label{functional_mechanim_sensitivity}
Let $f(\mathscr{D},w)$ and $f(\mathscr{D}',w)$ be the logistic regression on two neighboring datasets $\mathscr{D}$ and $\mathscr{D}'$, respectively, and denote their polynomial representations as $f(\mathscr{D},w)=\sum_{i=1}^n \sum_{j=1}^J \sum_{\phi \in \Phi_j}\lambda_{\phi t_{i}} \phi(w)$ and $f(\mathscr{D}',w)=\sum_{i=1}^n \sum_{j=1}^J \sum_{\phi \in \Phi_j}\lambda_{\phi t'_{i}} \phi(w)$.
Then, we have the following inequality
\begin{align*}
    \Delta_1 &= \sum_{j=1}^2 \sum_{\phi \in \Phi_j} \|\sum_{t_i \in \mathscr{D}} \lambda_{\phi t_{i}} - \sum_{t'_i \in \mathscr{D}'} \lambda_{\phi t'_{i}} \|_1\\
    &\leq 2 \max_{t}\sum_{j=1}^2\sum_{\phi \in \Phi_j}\|\lambda_{\phi t}\|_1 \leq \frac{d^2}{4}+d,
\end{align*}
where $t_i$, $t'_i$ or $t$ is an arbitrary tuple.
\end{lemma}
\subsection{Classification Fairness}
The goal of classification fairness is to find a classifier that minimizes the empirical loss while guaranteeing certain fairness requirements.
Many fairness definitions have been proposed for in the literature including mistreatment parity \cite{zafar2017fairness}, demographic parity \cite{pedreshi2008discrimination}, etc. 

Demographic parity, the most widely-used fairness definition in the classification fairness domain, requires the decision made by the classifier is not dependent on the protected attribute $z$, for instance, sex or race. 

\begin{definition}{(\bf{Demographic Parity in a Classifier)}} Given a classification model $\hat{y}=\rho(\x,w)$ and a labeled dataset $\mathscr{D}$, the property of demographic parity in a classifier is defined by $\Pr(\hat{y}=1|z=1)=\Pr(\hat{y}=1|z=0)$
where $z\in\{0,1\}$ is the protected attribute. 
\end{definition}
Moreover, demographic parity is quantified in terms of the risk difference (RD) \cite{pedreschi2012study}, i.e., the difference of the positive decision made in between the protected group and unprotected group. Thus, the risk difference produced by a classifier is defined as $RD=|\Pr(\hat{y}=1|z=1)-\Pr(\hat{y}=1|z=0)|.$

One of the in-processing methods, called decision boundary fairness \cite{zafar2015fairness}, to ensure classification fairness is to find a model parameter $w$ that minimizes the loss function $f(\mathscr{D},w)$ under a fairness constraint. Thus, the fair classification problem is formulated as follows,
\begin{align}\label{fairfunction}
    &\mathit{minimize}~~~ f(\mathscr{D},w)\nonumber\\
    &\mathit{subject}~ \mathit{to} ~~~g(\mathscr{D},w)\leq \tau, g(\mathscr{D},w)\geq -\tau,
\end{align}
where $g(\mathscr{D},w)$ is a constraint term, and $\tau$ is the threshold. For instance, Zafar et al. \cite{zafar2015fairness} have proposed to adopt the decision boundary covariance to define the fairness constraint, i.e., 
\begin{align}\label{decsionb}
    g(\mathscr{D},w)&=\mathbb{E}[(z-\Bar{z})d(\x,w)]-\mathbb{E}[z-\Bar{z}]d(\x,w)\nonumber\\
    &\varpropto \sum_{i=1}^n(z_i-\Bar{z})d(\x_i,w),
\end{align}
where $\{d(\x_i,w)\}_{i=1}^n$ is decision boundary, $\Bar{z}$ is the average of the protected attribute and $\mathbb{E}[z-\Bar{z}]=0$. For logistic regression classification models, the decision boundary is defined by $\x^Tw$. The decision boundary covariance \eqref{decsionb} then reduces to $g(\mathscr{D},w)=\sum_{i=1}^n(z_i-\Bar{z})\x_i^T w$.

\section{Differentially Private and Fair Classification}
In this section, we first present our approach PDFC to achieve fair logistic regression with $\epsilon$-differentially private guarantee. Then we propose a relaxed functional mechanism by injecting Gaussian noise instead of Laplace noise to provide $(\epsilon,\delta)$-differential privacy. By leveraging the relaxed functional mechanism, we will show that our second approach ADFC can jointly provide $(\epsilon,\delta)$-differential privacy and fairness. 
\subsection{Purely DP and Fair Classification}

In order to meet the requirements of $\epsilon$-differential privacy and fairness, motivated by \cite{xu2019achieving}, we consider to combine the functional mechanism and decision boundary fairness. We first consider to transform the constrained optimization problem \eqref{fairfunction} into unconstrained problem by treating the fairness constraint as a penalty term, where the fairness constraints are shifted to the original objective function $f(\mathscr{D},w)$. Then, we have the new objective function $\Tilde{f}_{\mathscr{D}}(w)$ defined as $\Tilde{f}(\mathscr{D},w)=f(\mathscr{D},w)+\alpha_1 |g(\mathscr{D},w)-\tau|$, where we consider $\alpha_1$ as a hyperparameter to optimize the trade-off between model utility and fairness. For convenience of discussion, we set $\tau=0$ and choose suitable values to make $\alpha_1= 1$. Note that our theoretical results still hold if we choose other values of $\alpha_1$ and $\tau$. By equation \eqref{decsionb}, we have 
\begin{align}
    \Tilde{f}(\mathscr{D},w)=& \sum_{i=1}^n[\log(1+exp(\x_i^T w))-y_i\x_i^T w]\nonumber\\
    &+ \left|\sum_{i=1}^n(z_i-\Bar{z})\x_i^T w \right|.
 \end{align}
To apply functional mechanism, we first write the approximate objective function $\Bar{f}(\mathscr{D},w)$ based on \eqref{ployre} as follows.
\begin{align}\label{approximate}
    \Bar{f}(\mathscr{D},w)&= \sum_{i=1}^n \sum_{j=0}^2 \frac{f_{1}^{(j)}(0)}{j!}(\x_i^T w)^j - \left(\sum_{i=1}^n y_i\x_i^T \right)w \nonumber \\
    &~~~+\left|\sum_{i=1}^n (z_i-\Bar{z})\x_i^T w\right| \nonumber\\
    &=\sum_{i=1}^n \sum_{j=0}^2 \sum_{\phi \in \Phi_j}\Bar{\lambda}_{\phi t_{i}} \phi(w),
\end{align}
where $\Bar{\lambda}_{\phi t_{i}}$ denotes the coefficient of $\phi(w)$ in the polynomial of $\Bar{f}(t_i,w)$ and $f_{1}(\cdot)=\log (1+\exp{(\cdot)})$.
\begin{algorithm}[!t]
\caption{Purely DP and Fair Classification (PDFC) }
\label{alg:PDFC}
\algsetup{indent=2em}
\begin{algorithmic}[1]
 
\STATE \textbf{Input:} Dataset $\mathscr{D}$; The objective function $f(\mathscr{D},w)$; The fairness constraint $g(\mathscr{D},w)$;
The privacy budget $\epsilon_s$ for unprotected attribute $x_s$; The privacy budget $\epsilon_n$ for other unprotected attributes $\{\x \setminus x_s\}$; $l_1$-sensitivity $\Delta_1$.
\STATE \textbf{Output:} $\hat{w}$, $\epsilon$.
\STATE Set the approximate function $\Bar{f}{(\mathscr{D},w)}$ by equation \eqref{approximate}.
\STATE Set two sets $\Phi_{s} =\{\}$, $\Phi_{n} =\{\}$. 
\FOR{$1\leq j \leq 2$}
\FOR{each $\phi \in \Phi_j $}
\IF{$\phi$ includes $w_s$ for a particular attribute $x_s$}
\STATE Put $\phi$ into $\Phi_{s}$.
\ELSE
\STATE Put $\phi$ into $\Phi_{n}$.
\ENDIF
\ENDFOR
\ENDFOR
\FOR{$1\leq j\leq 2$}
\FOR{each $\phi \in \Phi_j $}
\IF{$\phi \in \Phi_s$}
\STATE Set $\hat{\lambda}_{\phi}=\sum_{i=1}^n \Bar{\lambda}_{\phi t_i} + Lap(\Delta_1/(\epsilon_s))$.
\ELSE
\STATE Set $\hat{\lambda}_{\phi}=\sum_{i=1}^n\Bar{\lambda}_{\phi t_i} + Lap(\Delta_1/(\epsilon_n))$.
\ENDIF
\ENDFOR
\ENDFOR
\STATE Let $\hat{f}(\mathscr{D},w)=\sum_{j=1}^2\sum_{\phi \in \Phi_j}\hat{\lambda}_{\phi}\phi(w)$.
\STATE Compute $\hat{w}=\arg \min_{w}\hat{f}(\mathscr{D},w)$.
\STATE Compute $\epsilon=\epsilon_s /d + \epsilon_n (d-1)/d.$
\STATE \textbf{return:} $\hat{w}$, $\epsilon$.
\end{algorithmic}
\end{algorithm}

The attributes involving in the dataset may not be independent from each other, which means some unprotected attributes in $\x$ are quite correlated with the protected attribute $z$. For instance, the protected attribute, like gender, may be correlated with the attribute, marital status. Thus, to reduce the discrimination between the protected attribute $z$ and the labels $y$, it is important to weaken the correlation between these most correlated attributes and protected attribute $z$. However, it is often impossible to determine the degree of relation between an unprotected attribute and the protected attribute. Therefore, we randomly select an unprotected attribute $x_{s}$ and leverage functional mechanism to add noise with large scale to the corresponding polynomial coefficients of the monomials involving $w_s$. Interestingly, this approach not only helps to reduce the correlation between attributes $x_s$ and $z$, but also improve the privacy on attribute $x_s$ to prevent model inversion attacks, as shown in \cite{wang2015regression}. 

The key steps of PDFC are outlined in Algorithm \ref{alg:PDFC}. We first set two different privacy budgets, $\epsilon_s$ and $\epsilon_n$, for attribute $x_s$ and the rest of attributes $\{\x \setminus x_s\}$. Before injecting noise to the coefficients, all coefficients $\phi$ should be separated into two groups $\Phi_s$ and $\Phi_n$ by considering whether $w_s$ involves in the corresponding monomials (i.e., whether their the coefficients contain attribute $x_s$). 
We then add Laplace noises drawn from $Lap(\Delta_1/\epsilon_s)$ and $Lap(\Delta_1/\epsilon_n)$ to the coefficients of $\phi \in \Phi_s $ and $\phi \in \Phi_n $ respectively to reconstruct the differentially private objective function $\hat{f}({\mathscr{D},w})$, where $\Delta_1$ can be found in Lemma \ref{L1}. Finally, the differentially private model parameter $\hat{w}$ is obtained by minimizing $\hat{f}({\mathscr{D},w})$. Note that $\hat{w}$ also ensures classification fairness due to the objective function involving fairness constraint.


\begin{lemma} \label{L1}
Let $\mathscr{D}$ and $\mathscr{D}'$ be any two neighboring datasets differing in at most one tuple. Let $\Bar{f}(\mathscr{D},w)$ and $\Bar{f}(\mathscr{D}',w)$ be the approximate objective function on $\mathscr{D}$ and $\mathscr{D}'$, then we have the following inequality,
\begin{align*}
    \Delta_1 = \sum_{j=1}^2 \sum_{\phi \in \Phi_j} \|\sum_{i=1}^n\Bar{\lambda}_{\phi t_{i}} -\sum_{i=1}^n\Bar{\lambda}_{\phi t_{i}'} \|_1 \leq \frac{d^2}{4}+3d.
\end{align*}
\end{lemma}
The following theorem shows the privacy guarantee of PDFC.
\begin{theorem}
The output model parameter $\hat{w}$ in PDFC (Algorithm \ref{alg:PDFC}) preserves $\epsilon$-differential privacy, where $\epsilon = \frac{1}{d}\epsilon_s+ \frac{d-1}{d}\epsilon_n$.
\end{theorem}
\subsection{Approximately DP and Fair Classification}
We now focus on using the relaxed version of $\epsilon$-differential privacy, i.e., $(\epsilon,\delta)$-differential privacy to further improve the utility of differentially private and fair logistic regression. Hence, in order to satisfy $(\epsilon,\delta)$-differential privacy, we propose the relaxed functional mechanism by making use of Extended Gaussian mechanism. 
As shown in Theorem \ref{gaussin}, before applying Extended Gaussian mechanism, we first calculate the sensitivity of a query function, i.e., the objective function of logistic regression $f(\mathscr{D},w)=\sum_{i=1}^n[\log(1+exp(\x_i^T w))-y_i\x_i^T w]$, given in the following lemma.

\begin{lemma}[\bf{$l_2$-Sensitivity of Logistic Regression}]\label{relaxed_senstivity}
For polynomial representations of logistic regression, two $f(\mathscr{D},w)$ and $f(\mathscr{D}',w)$ given in Lemma \ref{functional_mechanim_sensitivity}, we have the following inequality
\begin{align*}
    \Delta_2 =\|\mathscr{A}_1-\mathscr{A}_2\|_2 \leq \sqrt{\frac{d^2}{16}+d},
\end{align*}
where we denote $\mathscr{A}_1=\left\{\sum_{i=1}^n {\lambda}_{\phi t_{i}} \right\}_{\phi \in \cup_{j=1}^J \Phi_{j}} $ and $\mathscr{A}_2=\left\{\sum_{i=1}^n {\lambda}_{\phi t_{i}'} \right\}_{\phi \in \cup_{j=1}^J \Phi_{j}} $ as the set of polynomial coefficients of $f(\mathscr{D},w)$ and $f(\mathscr{D}',w)$. And we denote $t_i$ or $t_i'$ as an arbitrary tuple.
\end{lemma}

We then perturb $f(\mathscr{D},w)$ by injecting Gaussian noise drawn from $\mathcal{N}(0,\sigma^2)$ with $\sigma = \frac{\sqrt{2}\Delta_2}{2\epsilon}(\sqrt{\log (\sqrt{\frac{2}{\pi}}\frac{1}{\delta}})+\sqrt{\log (\sqrt{\frac{2}{\pi}}\frac{1}{\delta})+\epsilon})$ into its polynomial coefficients, and obtain the differentially private model parameter $\hat{w}$ by minimizing the noisy function $\hat{f}(\mathscr{D},w)$, as shown in Algorithm \ref{realxed_functional}. Finally, we provide a privacy guarantee of proposed relaxed functional mechanism by the following theorem.

\begin{theorem}
The relaxed functional mechanism in Algorithm \ref{realxed_functional} guarantees $(\epsilon,\delta)$-differential privacy.
\end{theorem}

\begin{algorithm}[!t]
\caption{Relaxed Functional Mechanism }
\label{realxed_functional}
\algsetup{indent=2em}
\begin{algorithmic}[1]
\STATE \textbf{Input:} Dataset $\mathscr{D}$; The objective function ${f}{(\mathscr{D},w)}=\sum_{i=1}^n \sum_{j=1}^J \sum_{\phi \in \Phi_j}{\lambda}_{\phi t_{i}} \phi(w)$;
The privacy parameters $\epsilon,\delta$.
\STATE \textbf{Output:} $\hat{w}$
\STATE Set $\Delta_2$ according Lemma \ref{relaxed_senstivity}.
\FOR{$1\leq j\leq J$}
\FOR{each $\phi \in \Phi_j $}
\STATE Set ${\lambda}_{\phi}=\sum_{i=1}^n{\lambda}_{\phi t_i} + \mathcal{N}(0,\sigma^2)$, where $\sigma = \frac{\sqrt{2}\Delta_2}{2\epsilon}(\sqrt{\log (\sqrt{\frac{2}{\pi}}\frac{1}{\delta}})+\sqrt{\log (\sqrt{\frac{2}{\pi}}\frac{1}{\delta})+\epsilon})$.
\ENDFOR
\ENDFOR
\STATE Let $\hat{f}(\mathscr{D},w)=\sum_{j=1}^J\sum_{\phi \in \Phi_j}{\lambda}_{\phi}\phi(w)$.
\STATE Compute $\hat{w}=\arg \min_{w}\hat{f}(\mathscr{D},w)$.
\STATE \textbf{return:} $\hat{w}$.
\end{algorithmic}
\end{algorithm}
Our second approach called, ADFC, applies the relaxed functional mechanism into the objective function with decision boundary fairness constraint to enforce $(\epsilon,\delta)$-differential privacy and fairness. As shown in Algorithm \ref{delta_dp}, we first derive the polynomial representation $\Bar{f}(\mathscr{D},w)$ according to (\ref{approximate}), and employ random Gaussian noise to perturb the objective function $\Bar{f}(\mathscr{D},w)$, i.e., injecting Gaussian noise into its polynomial coefficients. Furthermore, we also allocate differential privacy parameters, $(\epsilon_s,\delta_s)$ and $(\epsilon_n,\delta_n)$ for a particular unprotected attribute $x_s$ and the rest of unprotected attributes $\{\x \setminus x_s\}$ to improve the privacy on attribute $x_s$ and reduce the correlation between attributes $x_s$ and $z$. Hence, we add random noise drawn from $\mathcal{N}(0, \sigma_s^2)$
to polynomial coefficients of $\phi \in \Phi_s$.
For polynomial coefficients in $\Phi_n$, we inject noise drawn from $\mathcal{N}(0, \sigma_n^2)$.

\begin{lemma} \label{L2}
Let $\mathscr{D}$ and $\mathscr{D}'$ be any two neighboring datasets differing in at most one tuple. Let $\Bar{f}(\mathscr{D},w)$ and $\Bar{f}(\mathscr{D}',w)$ be the approximate objective function on $\mathscr{D}$ and $\mathscr{D}'$, then we have the following inequality,
\begin{align*}
    \Delta_2' =\|\mathscr{A}_1'-\mathscr{A}_2'\|_2 \leq \sqrt{\frac{d^2}{16}+9d}.
\end{align*}
where we denote $\mathscr{A}_1'=\left\{\sum_{i=1}^n \bar{\lambda}_{\phi t_{i}} \right\}_{\phi \in \cup_{j=1}^2 \Phi_{j}} $ and $\mathscr{A}_2'=\left\{\sum_{i=1}^n \bar{\lambda}_{\phi t_{i}'} \right\}_{\phi \in \cup_{j=1}^2 \Phi_{j}} $ as the set of polynomial coefficients of $\Bar{f}(\mathscr{D},w)$ and $\Bar{f}(\mathscr{D}',w)$. And we denote $t_i$ or $t_i'$ as an arbitrary tuple.
\end{lemma}

\begin{algorithm}[!t]
\caption{Approximately DP and Fair Classification (ADFC) }
\label{delta_dp}
\algsetup{indent=2em}
\begin{algorithmic}[1]
\STATE \textbf{Input:} Dataset $\mathscr{D}$; The objective function $f(\mathscr{D},w)$; The fairness constraint $g(\mathscr{D},w)$;
The privacy parameters $\epsilon_s,\delta_s$ for unprotected attribute $x_s$; The privacy parameters  $\epsilon_n,\delta_n$ for other unprotected attributes $\{\x \setminus x_s\}$.
\STATE \textbf{Output:} $\hat{w}$, $\epsilon$ and $\delta$.
\STATE Set the approximate function $\Bar{f}{(\mathscr{D},w)}$ by equation \eqref{approximate}.
\STATE Set two sets $\Phi_{s} =\{\}$, $\Phi_{n} =\{\}$. 
\FOR{$1\leq j \leq 2$}
\FOR{each $\phi \in \Phi_j $}
\IF{$\phi$ includes $w_s$ for a particular attribute $x_s$}
\STATE Put $\phi$ into $\Phi_{s}$.
\ELSE
\STATE Put $\phi$ into $\Phi_{n}$.
\ENDIF
\ENDFOR
\ENDFOR
\STATE Set $l_2$-sensitivity $\Delta_2'$ by Lemma \ref{L2}.
\FOR{$1\leq j\leq 2$}
\FOR{each $\phi \in \Phi_j $}
\IF{$\phi \in \Phi_s$}
\STATE Set $\hat{\lambda}_{\phi}=\sum_{i=1}^n\bar{\lambda}_{\phi t_i} + \mathcal{N}(0,\sigma_s^2)$, where $\sigma_s = \frac{\sqrt{2}\Delta_2'}{2\epsilon_s}(\sqrt{\log (\sqrt{\frac{2}{\pi}}\frac{1}{\delta_s}})+\sqrt{\log (\sqrt{\frac{2}{\pi}}\frac{1}{\delta_s})+\epsilon_s})$.
\ELSE
\STATE Set $\hat{\lambda}_{\phi}=\sum_{i=1}^n\bar{\lambda}_{\phi t_i} + \mathcal{N}(0,\sigma_n^2)$, where $\sigma_n = \frac{\sqrt{2}\Delta_2'}{2\epsilon_n}(\sqrt{\log (\sqrt{\frac{2}{\pi}}\frac{1}{\delta_n}})+\sqrt{\log (\sqrt{\frac{2}{\pi}}\frac{1}{\delta_n})+\epsilon_n})$.
\ENDIF
\ENDFOR
\ENDFOR
\STATE Let $\hat{f}(\mathscr{D},w)=\sum_{j=1}^2\sum_{\phi \in \Phi_j}\hat{\lambda}_{\phi}\phi(w)$.
\STATE Compute $\hat{w}=\arg \min_{w}\hat{f}(\mathscr{D},w)$.
\STATE Compute $\epsilon = \frac{1}{d}\epsilon_s+ \frac{d-1}{d}\epsilon_n$ and $\delta = 1-(1-\delta_s)(1-\delta_n)$.
\STATE \textbf{return:} $\hat{w}$, $\epsilon$ and $\delta$.
\end{algorithmic}
\end{algorithm}
Finally, by minimizing the differentially private objective function $\hat{f}(\mathscr{D},w)$, we derive the model parameter $\hat{w}$, which achieves differential privacy and fairness at the same time. We now show that ADFC satisfies $(\epsilon,\delta)$-differential privacy in the following theroem.

\begin{theorem}
The output model parameter $\hat{w}$ in ADFC (Algorithm \ref{delta_dp}) guarantees $(\epsilon,\delta)$-differential privacy, where $\epsilon = \frac{1}{d}\epsilon_s+ \frac{d-1}{d}\epsilon_n$ and $\delta = 1-(1-\delta_s)(1-\delta_n)$.
\end{theorem}

\section{Performance Evaluation}
\subsection{Simulation Setup}

\begin{figure}
    \centering
    \includegraphics[width=0.27\textwidth]{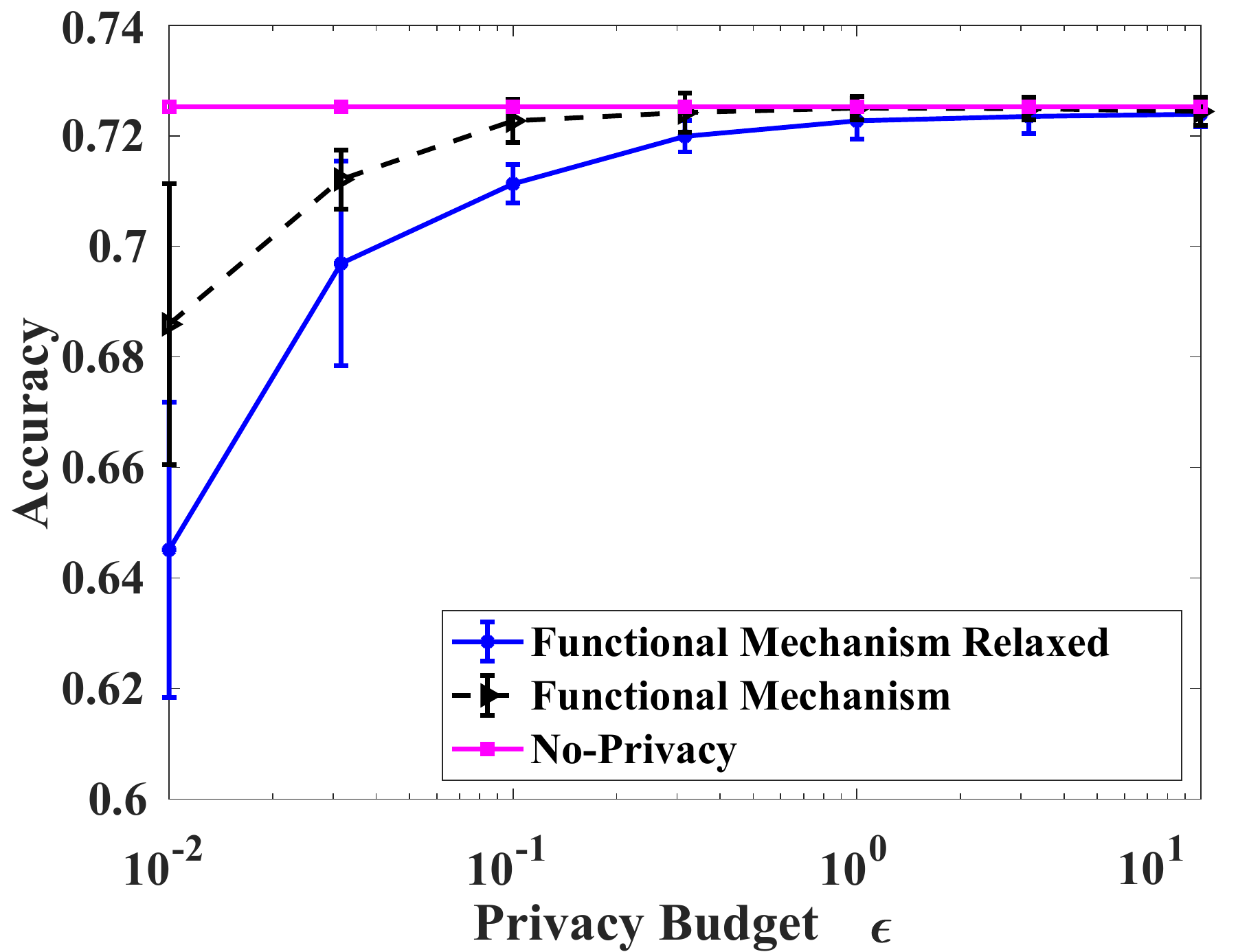}
    \caption{Compare accuracy under different privacy budgets on $\textit{US}.$ ($\delta=10^{-3}$)}
    \label{fig:relaxed}
\end{figure}

\begin{figure}
    \centering
    \includegraphics[width=0.27\textwidth]{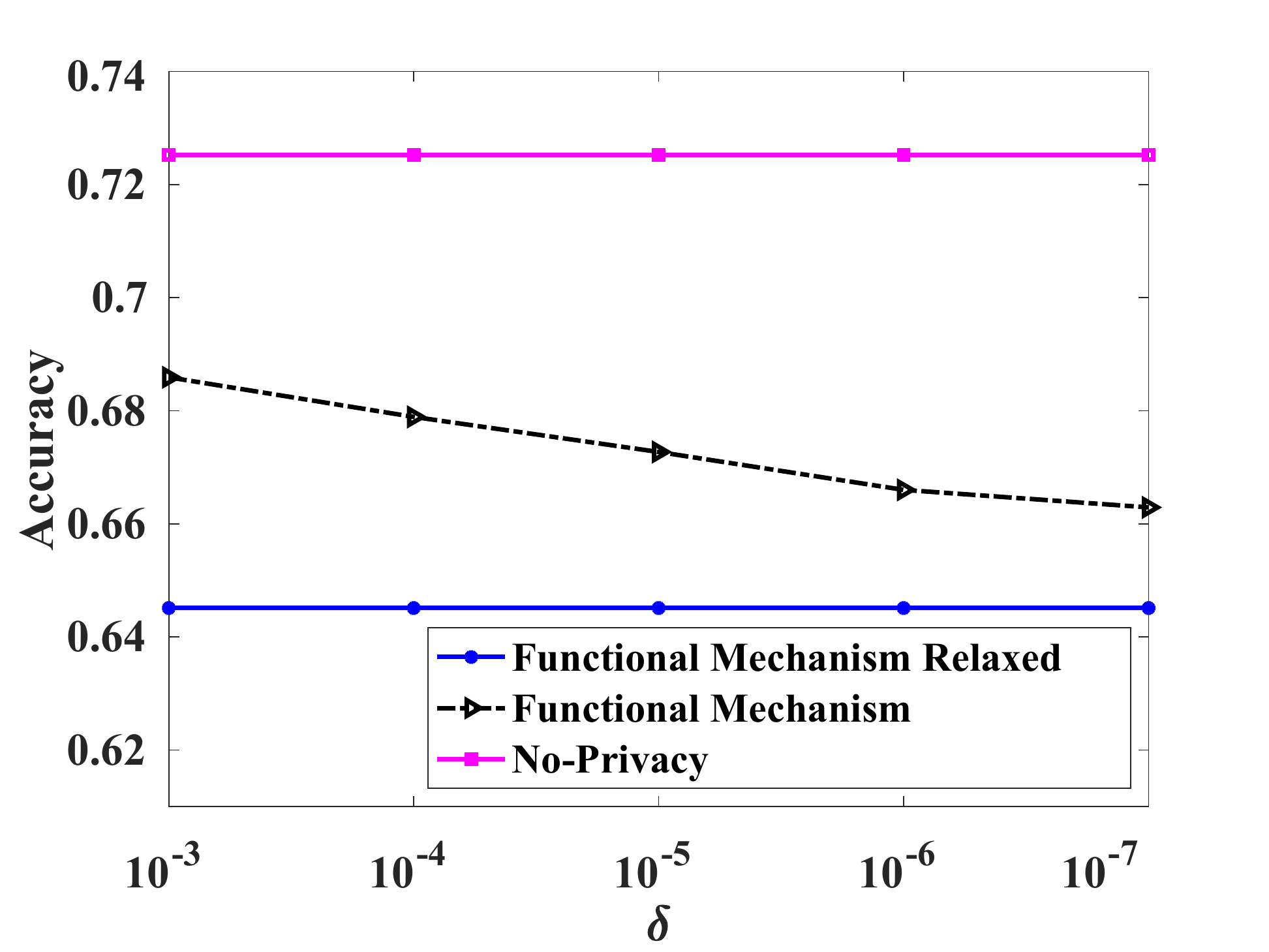}
    \caption{Compare accuracy under different values of $\delta$ on $\textit{US}$.}
    \label{fig:relaxed_delta}
\end{figure}

\begin{figure*} \centering
 \subfigure[\label{Peri_K_p}]
 {\includegraphics[width=2.3in]{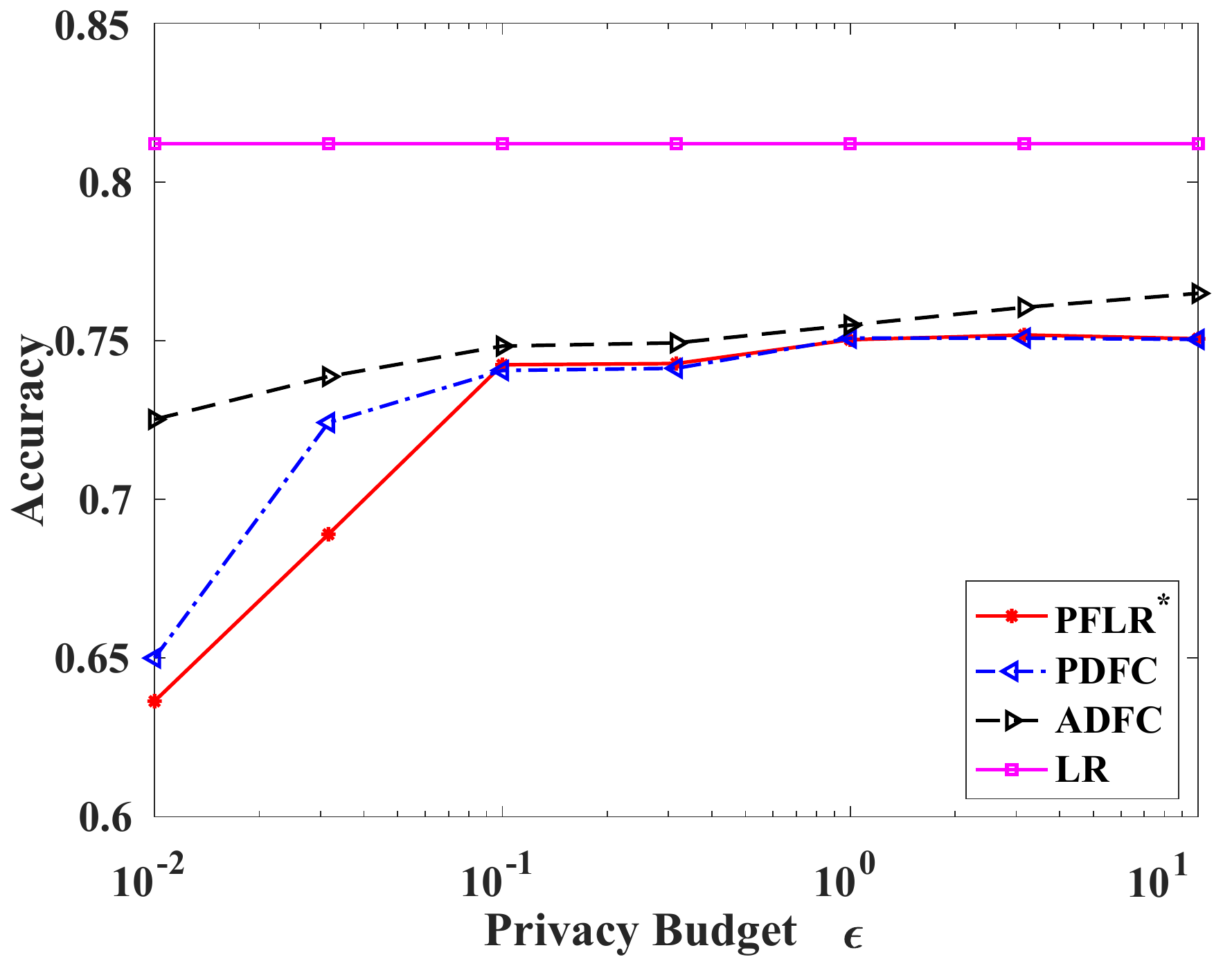}}
\subfigure[ \label{Peri_U}]
  {\includegraphics[width=2.3in]{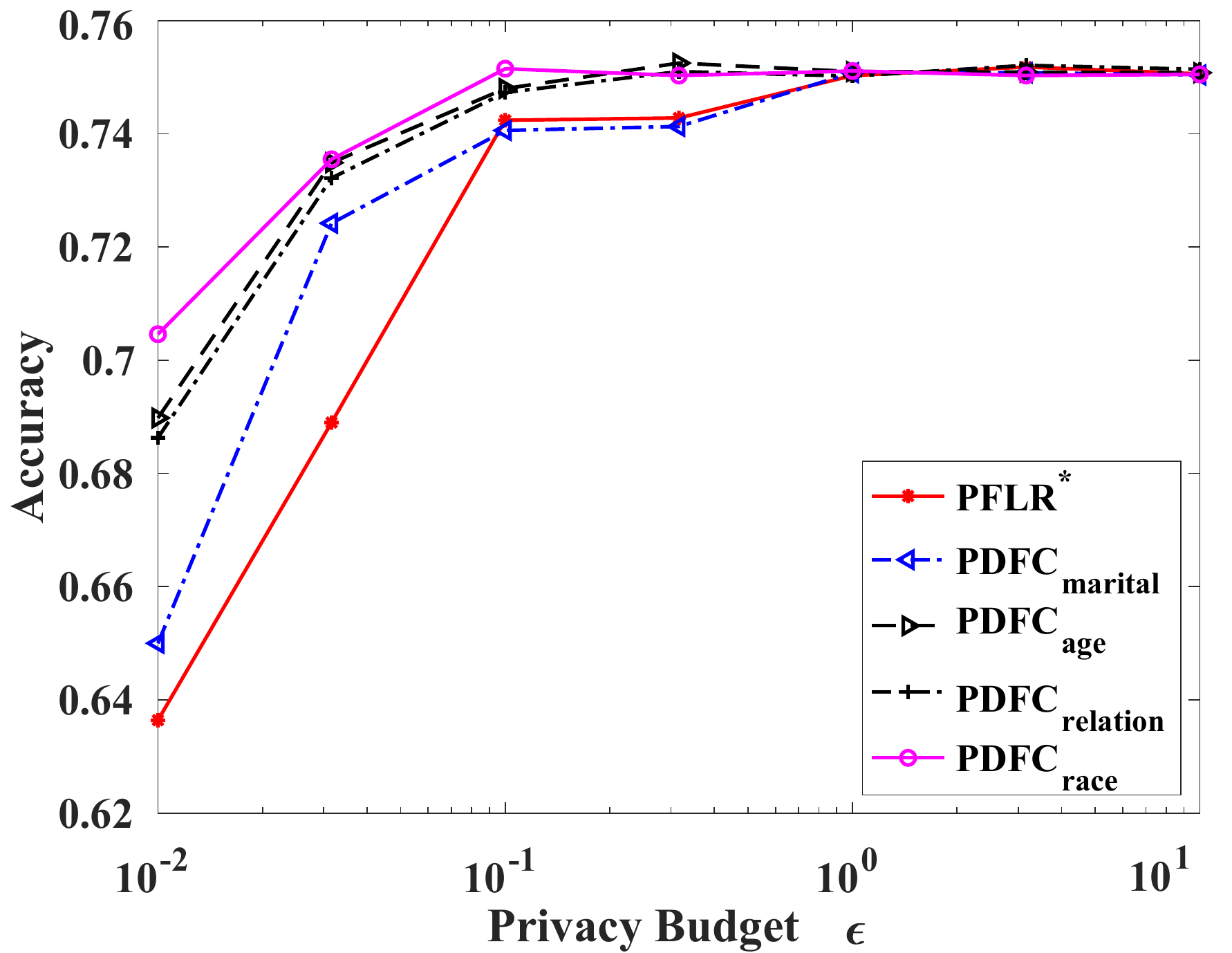}}
     \subfigure[\label{Peri_R_p}]
  {\includegraphics[width=2.3in]{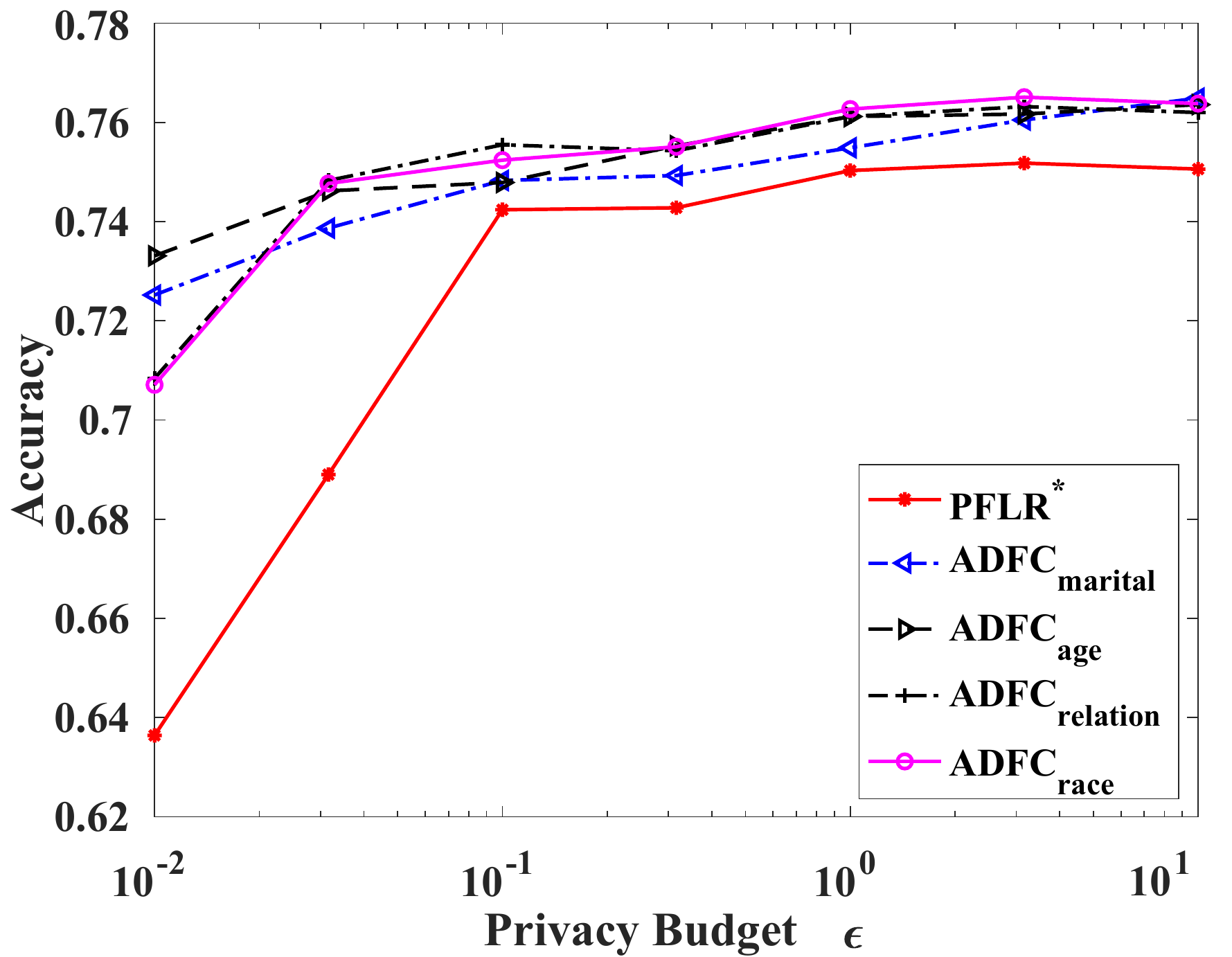}}
 \caption{Compare accuracy under different privacy budgets on $\textit{Adult}$ ($\delta=10^{-3}$).} \label{fig:comp_res}
\end{figure*}
\textbf{Data preprocessing} We evaluate the performance on two datasets, $\textit{Adult}$ dataset and $\textit{US}$ dataset. The $\textit{Adult}$ dataset from UCI Machine Learning Repository \cite{Dua:2017} contains information about 13 different features (e.g., work-class, education, race, age, sex, and so on) of 48,842 individuals. The label is to predict whether the annual income of those individuals is above 50K or not. The $\textit{US}$ dataset is from Integrated Public Use Microdata Series \cite{IPUM-S} and consists of 370,000 records of census microdata, which includes features like age, sex, education, family size, etc. The goal is to predict whether the income is over 25K a year. In both datasets, we consider sex as a binary protected attribute. 

\textbf{Baseline algorithms} In our experiments, we compare our approaches, PDFC, and ADFC against several baseline algorithms, namely, LR and PFLR*. LR is a logistic regression model.
PFLR* \cite{xu2019achieving} is a differentially private and fair logistic regression model that injects Laplace noise with shifted mean to the objective function of logistic regression with fairness constraint. Moreover, we compare our relaxed functional mechanism against the original functional mechanism proposed in \cite{zhang2012functional} and No-Privacy, which is the original functional mechanism without injecting any noise to the polynomial coefficients.

\textbf{Evaluation} The utility of algorithms is measured by $\textit{Accuracy}$, defined as follows, 
\begin{align*}
    Accuracy = \frac{Number~of~correct~predictions}{Total~number~of~predictions~made},
\end{align*}
which demonstrates the quality of a classifier. The fairness of classification models is qualified by $\textit{ risk difference (RD)}$
\begin{align*}
    RD=|\Pr(\hat{y}=1|z=1)-\Pr(\hat{y}=1|z=0)|,
\end{align*}
where $z$ is the protected attribute. We consider a random 80-20 training-testing split and conduct 10 independent runs of algorithms. We then record the mean values and standard deviation values of $\textit{Accuracy}$ and $\textit{RD}$ on the testing dataset. For the parameters of differential privacy, we consider $\epsilon=\{10^{-2},10^{-1.5},10^{-1},10^{0},10^{0.5},10^{1}\}$, and $\delta =\{ 10^{-3}, 10^{-4},10^{-5},10^{-6},10^{-7}\}$.

\subsection{Results and Analysis}

In Figure \ref{fig:relaxed}, we show the 
accuracy of each algorithm, functional mechanism, relaxed functional mechanism and No-Privacy, as a function of the privacy budget with fixed $\delta = 10^{-3}$. We can see that the accuracy of No-Privacy remains unchanged for all values of $\epsilon$, as it does not provide any differential privacy guarantee. Our relaxed functional mechanism exhibits quite higher accuracy than functional mechanism in high privacy regime, and the accuracy of relaxed functional mechanism is the same as No-Privacy baseline when $\epsilon> 10^{-1}$. 
Figure \ref{fig:relaxed_delta} studies the accuracy of each algorithm under different values of $\delta$ with fixed $\epsilon=10^{-2}$. 
Relaxed functional mechanism incurs lower accuracy when $\delta$ decreases, as a smaller $\delta$ requires a larger scale of noise to be injected in the objective function. But the accuracy of functional mechanism remains considerably lower than relaxed functional mechanism in all cases. 

Figure \ref{Peri_K_p} studies the accuracy comparison among PFLR*, LR, PDFC and ADFC on $\textit{Adult}$ dataset with the particular unprotected attribute $x_s$ denoted by marital status. We can observe that ADFC continuously achieves better accuracy than PFLR* in all privacy regime, and PDFC only outperforms PFLR* when $\epsilon$ is small. We also evaluate the effect of choosing different attributes as $x_s$ by performing experiments on $\textit{Adult}$ dataset. As shown in Figure \ref{Peri_U} and Figure \ref{Peri_R_p}, 
choosing different attributes, marital status, age, relation and race, has different effects on the accuracy of PDFC and ADFC. However, PDFC and ADFC still outperform PFLR* under varying values of $\epsilon$. As expected, as the value
of $\epsilon$ increases, the accuracy of each algorithm becomes higher in above three figures.

Table \ref{tab_result} shows how different privacy budgets affect the risk difference of LR, PFLR*, PDFC and ADFC on two datasets. Note that we consider the attribute $x_s$ as race on $\textit{Adult}$ dataset, and work on $\textit{US}$ dataset. It is clear that PDFC and ADFC produce less risk difference compared to PFLR* in most cases of $\epsilon$. The key reason is that adding different amounts of noise regarding different attributes indeed reduces the correlation between unprotected attributes and protected attributes.

\begin{table}[ht]
\renewcommand{\arraystretch}{1}
\centering
\caption{Risk difference with different privacy budgets $\epsilon$ on \textit{US} dataset  ($\delta=10^{-3}$).}
\label{tab_result}
\scalebox{0.683}{
\begin{tabular}{|c|c|c|c|c|c|}
\hline
Data & $\epsilon$ &LR & PFLR* & PDFC &ADFC\\
\hline
\multicolumn{1}{|c|}{\multirow{7}*{$\textit{US}$}} & $0.01 $& $0.191 \pm 0.014$ &$0.037 \pm 0.038$ & $0.003 \pm 0.034$ &$0.004 \pm 0.007$ \\
\cline{2-6}
\multicolumn{1}{|c|}{~} & $0.1$ &$0.191 \pm 0.014$& $0.078 \pm 0.021$& $0.001\pm0.006$&$0.008\pm 0.003$\\ 
\cline{2-6}
 \multicolumn{1}{|c|}{~}& $1$ &$0.191 \pm 0.014$& $0.069\pm 0.007$&$0.022 \pm 0.047$&$0.031\pm 0.004$ \\
\cline{2-6}
\multicolumn{1}{|c|}{~} & $10$ & $0.191 \pm 0.014$& $0.067\pm0.003$& $0.022\pm 0.031$&$0.045\pm0.002$\\
\hline
\end{tabular}
}
\end{table}

\begin{table}[ht]
\renewcommand{\arraystretch}{1}
\centering
\caption{Risk difference with different privacy budgets $\epsilon$ on two datasets ($\delta=10^{-3}$).}
\label{tab_result}
\scalebox{0.683}{
\begin{tabular}{|c|c|c|c|c|c|}
\hline
Data & $\epsilon$ &LR & PFLR* & PDFC &ADFC\\
\hline
\multicolumn{1}{|c|}{\multirow{7}*{$\textit{Adult}$}} & $0.01 $& $0.187  \pm 0.049$ &$0.045 \pm 0.095$ & $0.048 \pm 0.108$ &$0.146 \pm 0.131$ \\
\cline{2-6}
\multicolumn{1}{|c|}{~} & $0.1$ &$0.187  \pm 0.049$& $0.004 \pm 0.009$& $0.005\pm0.022$&$0.068\pm 0.028$\\ 
\cline{2-6}
 \multicolumn{1}{|c|}{~}& $1$ &$0.187  \pm 0.049$& $0.022\pm 0.088$&$0.002 \pm 0.011$&$0.045\pm 0.027$ \\
\cline{2-6}
\multicolumn{1}{|c|}{~} & $10$ & $0.187  \pm 0.049$& $0.003\pm0.001$& $0.035\pm0.041$&$0.019\pm0.003$\\
\hline
\multicolumn{1}{|c|}{\multirow{7}*{$\textit{US}$}} & $0.01 $& $0.191 \pm 0.014$ &$0.037 \pm 0.038$ & $0.003 \pm 0.034$ &$0.004 \pm 0.007$ \\
\cline{2-6}
\multicolumn{1}{|c|}{~} & $0.1$ &$0.191 \pm 0.014$& $0.078 \pm 0.021$& $0.001\pm0.006$&$0.008\pm 0.003$\\ 
\cline{2-6}
 \multicolumn{1}{|c|}{~}& $1$ &$0.191 \pm 0.014$& $0.069\pm 0.007$&$0.022 \pm 0.047$&$0.031\pm 0.004$ \\
\cline{2-6}
\multicolumn{1}{|c|}{~} & $10$ & $0.191 \pm 0.014$& $0.067\pm0.003$& $0.022\pm 0.031$&$0.045\pm0.002$\\
\hline
\end{tabular}
}
\end{table}

\section{Conclusion}

In this paper, we have introduced two approaches, PDFC and ADFC, to address the discrimination and privacy concerns in logistic regression classification. Different from existing techniques, in both approaches, we consider leveraging functional mechanism to the objective function with decision boundary fairness constraints, and adding noise with different magnitudes into the coefficients of different attributes to further reduce the discrimination and improve the privacy protection. Moreover, for ADFC, we utilize the proposed relaxed functional mechanism that is built upon Extended Gaussian mechanism, to further improve the model accuracy. By performing extensive empirical comparisons with state-of-the-art methods for differentially private and fair classification, we demonstrated the effectiveness of proposed approaches.

\section{Acknowledgments}
The work of J. Ding, X. Zhang, and M. Pan was supported in part by the U.S. National Science Foundation under grants US CNS-1350230 (CAREER), CNS-1646607, CNS-1702850, and CNS-1801925. The work of X. Li was supported in part by the Programs of NSFC under Grant 61762030, in part by the Guangxi Natural Science Foundation under Grant 2018GXNSFDA281013, and in part by the Key Science and Technology Project of Guangxi under Grant AA18242021.

\bibliography{6057-ding}
\bibliographystyle{aaai}
\addtocounter{lemma}{1}
\begin{lemma} \label{L1}
Let $\mathscr{D}$ and $\mathscr{D}'$ be any two neighboring datasets differing in at most one tuple. Let $\Bar{f}(\mathscr{D},w)$ and $\Bar{f}(\mathscr{D}',w)$ be the approximate objective function on $\mathscr{D}$ and $\mathscr{D}'$, then we have the following inequality,
\begin{align*}
    \Delta_1 = \sum_{j=1}^2 \sum_{\phi \in \Phi_j} \|\sum_{i=1}^n\Bar{\lambda}_{\phi t_{i}} -\sum_{i=1}^n\Bar{\lambda}_{\phi t_{i}'} \|_1 \leq \frac{d^2}{4}+3d.
\end{align*}
\end{lemma}
\begin{proof}
Assume that $\mathscr{D}$ and $\mathscr{D}'$ differ in the last tuple $t_n$ and $t'_n$. We have that
\begin{align}
    &\Delta_1 = \sum_{j=1}^2 \sum_{\phi \in \Phi_j} \|\sum_{i=1}^n\Bar{\lambda}_{\phi t_{i}} -\sum_{i=1}^n\Bar{\lambda}_{\phi t_{i}'} \|_1 \nonumber\\
    &= \sum_{j=1}^2 \sum_{\phi \in \Phi_j} \|\Bar{\lambda}_{\phi t_{n}} -\Bar{\lambda}_{\phi t_{n}'} \|_1 \nonumber\\
    &\leq 2 \max_{t=(\x,y)}\sum_{j=1}^2 \sum_{\phi \in \Phi_j}  \|\Bar{\lambda}_{\phi t}\|_1 \nonumber\\
    & \leq 2 \max_{t=(\x,y)}    (\frac{1}{2}-y_i + |z_i-\Bar{z}|)\sum_{e=1}^d x_{(e)} + \frac{1}{8}\sum_{1\leq e,l \leq d} x_{(e)}x_{(l)} \nonumber\\
    &\leq \frac{d^2}{4}+3d\nonumber
\end{align}
where $x_{(e)}$ represents $e$-th element in feature vector $\x$.
\end{proof}

\addtocounter{theorem}{2}
\begin{theorem}
The output model parameter $\hat{w}$ in PDFC (Algorithm 1) guarantees $\epsilon$-differential privacy, where $\epsilon = \frac{1}{d}\epsilon_s+ \frac{d-1}{d}\epsilon_n$.
\end{theorem}
\begin{proof}
We assume there are two neighboring datasets $\mathscr{D}$ and $\mathscr{D}'$ that differ in the last tuple $t_n$ and $t_n'$. As shown in the Algorithm 1, all polynomial coefficients $\phi$ are divided into two subsets $\Phi_s$ and $\phi_n$ in view of whether they include sensitive attribute $x_s$ or not. After adding Laplace noise, we have
\begin{align*}
    &\Pr \left(\hat{f}(\mathscr{D},w)\right)\\ &= \prod_{\phi\in\Phi_{s} }\exp{\bigg(\frac{\epsilon_s \|\sum_{i=1}^n\Bar{\lambda}_{\phi t_{i}}-\hat{\lambda}_{\phi}\|_1}{\Delta_1}}\bigg)\\
    &~~~~~\prod_{\phi\in\Phi_{n} }\exp{\bigg(\frac{\epsilon_n \|\sum_{i=1}^n\bar{\lambda}_{\phi t_{i}}-\hat{\lambda}_{\phi}\|_1}{\Delta_1}\bigg)}
\end{align*}
Then, the following inequality holds
\begin{align*}
    &\frac{\Pr\left(\hat{f}(\mathscr{D},w)\right)}{\Pr \left(\hat{f}(\mathscr{D}',w)\right)} \\
    &\leq \prod_{\phi\in\Phi_{s} }
    \exp{\bigg(\frac{\epsilon_s
    \|\sum_{i=1}^n \Bar{\lambda}_{\phi t_{i}}
    -\sum_{i=1}^n\Bar{\lambda}_{\phi t_{i}'}\|_1 }{\Delta_1}\bigg)}\\
    &~~~~~\prod_{\phi\in\Phi_{n} }
    \exp{\bigg(\frac{\epsilon_n
    \|\sum_{i=1}^n \Bar{\lambda}_{\phi t_{i}}
    -\sum_{i=1}^n\Bar{\lambda}_{\phi t_{i}'}\|_1 }{\Delta_1}\bigg)}
    \\
    &\leq \prod_{\phi\in\Phi_{s} }
    \exp{\bigg(\frac{\epsilon_s
    \|\Bar{\lambda}_{\phi t_{n}}
    -\Bar{\lambda}_{\phi t_{n}'}\|_1 }{\Delta_1}\bigg)}\\
    &~~~~~\prod_{\phi\in\Phi_{n} }
    \exp{\bigg(\frac{\epsilon_n
    \|\Bar{\lambda}_{\phi t_{n}}
    -\Bar{\lambda}_{\phi t_{n}'}\|_1 }{\Delta_1}\bigg)}\\
    &\leq \prod_{\phi\in\Phi_{s} } \exp{(\frac{\epsilon_s}{\Delta_1} 2 \max_{t} \|\lambda_{\phi t}\|_1 )} \prod_{\phi\in\Phi_{n} } \exp{(\frac{\epsilon_n}{\Delta_1} 2 \max_{t} \|\lambda_{\phi t}\|_1)}\\
    &= \exp{(\epsilon_s /d + \epsilon_n (d-1)/d)}\\
    &=\exp{(\epsilon)}.
\end{align*}
In the last second equality, we directly adopt the result in \cite{wang2015regression}.

\end{proof}

\begin{lemma}[\bf{$l_2$-Sensitivity of Logistic Regression}]\label{relaxed_senstivity}
For polynomial representations of logistic regression, two $f(\mathscr{D},w)$ and $f(\mathscr{D}',w)$ given in Lemma 1, we have the following inequality
\begin{align*}
    \Delta_2 =\|\mathscr{A}_1-\mathscr{A}_2\|_2 \leq \sqrt{\frac{d^2}{16}+d},
\end{align*}
where we denote $\mathscr{A}_1=\left\{\sum_{i=1}^n {\lambda}_{\phi t_{i}} \right\}_{\phi \in \cup_{j=1}^J \Phi_{j}} $ and $\mathscr{A}_2=\left\{\sum_{i=1}^n {\lambda}_{\phi t_{i}'} \right\}_{\phi \in \cup_{j=1}^J \Phi_{j}} $ as the set of polynomial coefficients of $f(\mathscr{D},w)$ and $f(\mathscr{D}',w)$. And we denote $t_i$ or $t_i'$ as an arbitrary tuple.
\end{lemma}
\begin{proof}
Assume that $\mathscr{D}$ and $\mathscr{D}'$ differ in the last tuple $t_n$ and $t'_n$. For logistic regression, we have
\begin{align*}
    f(\mathscr{D},w) &= \sum_{i=1}^n \sum_{j=0}^2 \frac{f_{1}^{(j)}(0)}{j!}(\x_i^T w)^j - \left(\sum_{i=1}^n y_i\x_i^T \right)w \nonumber \\
    &=\sum_{i=1}^n \sum_{j=0}^2 \sum_{\phi \in \Phi_j} {\lambda}_{\phi t_{i}} \phi(w)
\end{align*}
where we have 
\begin{align*}
    \{{\lambda}_{\phi t_{i}}\}_{\phi \in \Phi_1} &=:\lambda_{1 t_i} =   \frac{f_{1}^{(1)}(0)}{1!}\x_i - y_i\x_i= (\frac{1}{2}-y_i)\x_i,\\
    \{{\lambda}_{\phi t_{i}}\}_{\phi \in \Phi_2} &=:\lambda_{2 t_i} =   \frac{f_{1}^{(2)}(0)}{2!}\x_i^2=\frac{1}{8}\x_i^2.
\end{align*}
Denote $\mathscr{A}_1=\left\{\sum_{i=1}^n {\lambda}_{\phi t_{i}} \right\}_{\phi \in \cup_{j=1}^2 \Phi_{j}} $ and $\mathscr{A}_2=\left\{\sum_{i=1}^n {\lambda}_{\phi t_{i}'} \right\}_{\phi \in \cup_{j=1}^2 \Phi_{j}} $ as the set of polynomial coefficients of $f(\mathscr{D},w)$ and $f(\mathscr{D}',w)$, and $\mathscr{E}=\begin{pmatrix}
    (\frac{1}{2}-y)x_{(1)} \\ \cdots \\ (\frac{1}{2}-y)x_{(d)} \\ \frac{1}{8}x_{(1)}x_{(1)}\\\cdots\\\frac{1}{8}x_{(d)}x_{(d)}
    \end{pmatrix}_{(d+d^2)\times 1}$, where $x_{(e)}$ represents $e$-th element in feature vector $\x$.
    
Then, we have 
\begin{align*}
    \Delta_2 &=\|\mathscr{A}_1-\mathscr{A}_2\|_2\\
    &= \|\{\sum_{i=1}^n {\lambda}_{\phi t_{i}}-\sum_{i=1}^n {\lambda}_{\phi t_{i}'}\}_{\phi \in \cup_{j=1}^2 \Phi_{j}}\|_2\\
    &=\|\{ {\lambda}_{\phi t_{n}}- {\lambda}_{\phi t_{n}'}\}_{\phi \in \cup_{j=1}^2 \Phi_{j}}\|_2\\
    &\leq 2 \max_{t=(\x,y)} \|\mathscr{E}\|_2\\
    &= 2 \max_{t=(\x,y)} \sqrt{\sum_{j=1}^d((\frac{1}{2}-y)x_j)^2+ \sum_{1\leq e,l \leq d} (\frac{1}{8}x_{(e)}x_{(l)})^2}\\
    &= \sqrt{\frac{d^2}{16}+d},
\end{align*}
where $t$ is an arbitrary tuple.
\end{proof}

\begin{theorem}
The relaxed functional mechanism in Algorithm 2 guarantees $(\epsilon,\delta)$-differential privacy.
\end{theorem}
\begin{proof}
Assume that the neighboring datasets $\mathscr{D}$ and $\mathscr{D}'$ differ in the last tuple $t_n$ and $t'_n$.
\begin{align*}
   &\left|\log \frac{\Pr\left(\hat{f}(\mathscr{D},w)\right)}{\Pr \left(\hat{f}(\mathscr{D}',w)\right)} \right|\\
   &=\left|\log \frac{\prod_{j=1}^J\prod_{\phi\in\Phi_{j} }
    \exp{\left(-\frac{1
    }{2\sigma^2}\left(\sum_{i=1}^n \Bar{\lambda}_{\phi t_{i}}
    -\hat{\lambda}_{\phi}\right)^2 \right)}}{\prod_{j=1}^J\prod_{\phi\in\Phi_{j} }
    \exp{\left(-\frac{1
    }{2\sigma^2}\left(\sum_{i=1}^n \Bar{\lambda}_{\phi t_{i}'}
    -\hat{\lambda}_{\phi}\right)^2 \right)}} \right|\\
    &= \frac{1}{2\sigma^2}\left|  \sum_{j=1}^J\sum_{\phi\in\Phi_{j}}\left(\sum_{i=1}^n \Bar{\lambda}_{\phi t_{i}}
    -\hat{\lambda}_{\phi}\right)^2 \right. \\
    &~~~~~~~~~~~~~~~~~~~~~~~~\left. - \left(\sum_{i=1}^n \Bar{\lambda}_{\phi t_{i}'}
    -\hat{\lambda}_{\phi}\right)^2 \right|\\
    &= \frac{1}{2\sigma^2}\left|\|\mathscr{A}\|_2^2-\|\mathscr{A}+\mathscr{B}\|_2^2\right|,
\end{align*}
where $\mathscr{A}=\left\{\sum_{i=1}^n \Bar{\lambda}_{\phi t_{i}}
    -\hat{\lambda}_{\phi} \right\}_{\phi \in \cup_{j=1}^J \Phi_{j}}$ and $\mathscr{B}=\left\{\sum_{i=1}^n \Bar{\lambda}_{\phi t_{i}'}-\sum_{i=1}^n \Bar{\lambda}_{\phi t_{i}}\right\}_{\phi \in \cup_{j=1}^J \Phi_{j}} $.
    
    We know the fact that the distribution of a spherically symmetric normal is not dependent of the orthogonal basis where its constituent normals are drawn. Thus, we work in a basis aligned with $\mathscr{B}$.
    Fix such a basis $\mathscr{C}_1,\cdots, \mathscr{C}_{\left|\cup_{j=1}^J \Phi_{j}\right|}$ and draw $\mathscr{A}$ by first drawing signed lengths $\mathscr{V}_{\phi}\sim \mathcal{N}(0,\sigma^2)$ for $\phi \in \cup_{j=1}^J \Phi_{j}$, then let $\mathscr{A}_{\phi} = \mathscr{V}_{\phi}\mathscr{C}_{\phi}$ and $\mathscr{A}=\sum_{\phi \in \cup_{j=1}^J \Phi_{j}}\mathscr{A}_{\phi}$. Without loss of generality, we assume that $\mathscr{C}_1$ is parallel to $\mathscr{B}$. Based on the triangle with the base $\mathscr{B}+\mathscr{A}_1$ and the edge $\sum_{\phi =2}^{\left|\cup_{j=1}^J \Phi_{j}\right|}\mathscr{A}_{\phi}$ orthogonal to $\mathscr{B}$, apparently, we have $\|\mathscr{A}+\mathscr{B}\|_2^2-\|\mathscr{A}\|_2^2 = \|\mathscr{B}\|_2^2+ 2\mathscr{C}_1\|\mathscr{B}\|_2$. Since $\|\mathscr{B}\|_2\leq \Delta_2$, we have
    $\left|\log\left(\Pr\left(\hat{f}(\mathscr{D},w)\right)/\Pr \left(\hat{f}(\mathscr{D}',w)\right)\right)\right|\leq \frac{1}{2\sigma^2}\left|\Delta_2^2+2|\mathscr{V}_1|\Delta_2\right|$.
    When $|\mathscr{V}_1|\leq \frac{1}{2}(2\sigma^2\epsilon-1)$, the privacy loss is bounded by $\epsilon(\epsilon>0)$, i.e., $\left|\log\left(\Pr\left(\hat{f}(\mathscr{D},w)\right)/\Pr \left(\hat{f}(\mathscr{D}',w)\right)\right)\right| \leq \epsilon$. Next, we need to prove that the privacy loss is bounded by $\epsilon$ with probability at least $1-\delta$, which requires $\Pr\left(\mathscr{V}_1 >\frac{1}{2}(2\sigma^2\epsilon-1) \right)\leq \delta/2$. Now we use the tail bound of $\mathscr{V}_{1}\sim \mathcal{N}(0,\sigma^2)$, we have 
    \begin{align*}
        \Pr\left(\mathscr{V}_1 > r \right)\leq \frac{\sigma}{\sqrt{2}r}\exp(-\frac{r^2}{2\sigma^2}).
    \end{align*}
    By letting $r =\frac{1}{2}(2\sigma^2\epsilon-1) $ in the above inequality, we have
    \begin{align*}
         \Pr\left(\mathscr{V}_1 > \frac{1}{2}(2\sigma^2\epsilon-1) \right)\leq \frac{\sqrt{2}\sigma}{2\sigma^2\epsilon-1}\exp\left(-\frac{1}{2}\left(\frac{2\sigma^2\epsilon-1}{2\sigma}\right)^2\right)
    \end{align*}
    When $\sigma\geq\frac{\sqrt{2}\Delta_2}{2\epsilon}(\sqrt{\log (\sqrt{\frac{2}{\pi}}\frac{1}{\delta}})+\sqrt{\log (\sqrt{\frac{2}{\pi}}\frac{1}{\delta})+\epsilon})$, $\epsilon>0$ and $\delta$ is very small, we have 
    \begin{align*}
        \Pr\left(\mathscr{V}_1 >\frac{1}{2}(2\sigma^2\epsilon-1) \right)\leq \delta/2.
    \end{align*}
We then can easily prove that 
\begin{align*}
    \Pr\left(|\mathscr{V}_1| \leq\frac{1}{2}(2\sigma^2\epsilon-1) \right) \geq 1-\delta.
\end{align*}
Based on the proof above, we know that the privacy loss $\left|\log\left(\Pr\left(\hat{f}(\mathscr{D},w)\right)/\Pr \left(\hat{f}(\mathscr{D}',w)\right)\right)\right|$ is bounded by $\epsilon$ with probability at least $1-\delta$, which represents the the computation of $\hat{f}(\mathscr{D},w)$ satisfies $(\epsilon,\delta)$-differential privacy. Therefore, Algorithm 2 satisfies $(\epsilon,\delta)$-differential privacy.
\end{proof}

\begin{lemma} \label{L2}
Let $\mathscr{D}$ and $\mathscr{D}'$ be any two neighboring datasets differing in at most one tuple. Let $\Bar{f}(\mathscr{D},w)$ and $\Bar{f}(\mathscr{D}',w)$ be the approximate objective function on $\mathscr{D}$ and $\mathscr{D}'$, then we have the following inequality,
\begin{align*}
    \Delta_2' =\|\mathscr{A}_1'-\mathscr{A}_2'\|_2 \leq \sqrt{\frac{d^2}{16}+9d}.
\end{align*}
where we denote $\mathscr{A}_1'=\left\{\sum_{i=1}^n \bar{\lambda}_{\phi t_{i}} \right\}_{\phi \in \cup_{j=1}^2 \Phi_{j}} $ and $\mathscr{A}_2'=\left\{\sum_{i=1}^n \bar{\lambda}_{\phi t_{i}'} \right\}_{\phi \in \cup_{j=1}^2 \Phi_{j}} $ as the set of polynomial coefficients of $\Bar{f}(\mathscr{D},w)$ and $\Bar{f}(\mathscr{D}',w)$. And we denote $t_i$ or $t_i'$ as an arbitrary tuple.
\end{lemma}
\begin{proof}
Assume that $\mathscr{D}$ and $\mathscr{D}'$ differ in the last tuple $t_n$ and $t'_n$. For objective $\Bar{f}(\mathscr{D},w)$, we have
\begin{align*}
    \Bar{f}(\mathscr{D},w)&= \sum_{i=1}^n \sum_{j=0}^2 \frac{f_{1}^{(j)}(0)}{j!}(\x_i^T w)^j - \left(\sum_{i=1}^n y_i\x_i^T \right)w \nonumber \\
    &~~~+\left|\sum_{i=1}^n (z_i-\Bar{z})\x_i^T w\right| \nonumber\\
    &=\sum_{i=1}^n \sum_{j=0}^2 \sum_{\phi \in \Phi_j}\Bar{\lambda}_{\phi t_{i}} \phi(w),
\end{align*}
where we have 
\begin{align*}
    \{{\Bar{\lambda}}_{\phi t_{i}}\}_{\phi \in \Phi_1} &=:\Bar\lambda_{1 t_i} = (\frac{1}{2}-y_i+|z_i-\Bar{z}|)\x_i,\\
    \{{\Bar\lambda}_{\phi t_{i}}\}_{\phi \in \Phi_2} &=:\Bar\lambda_{2 t_i} =\frac{1}{8}\x_i^2.
\end{align*}
Denote $\mathscr{A}_1'=\left\{\sum_{i=1}^n \bar{\lambda}_{\phi t_{i}} \right\}_{\phi \in \cup_{j=1}^2 \Phi_{j}} $ and $\mathscr{A}_2'=\left\{\sum_{i=1}^n \bar{\lambda}_{\phi t_{i}'} \right\}_{\phi \in \cup_{j=1}^2 \Phi_{j}} $ as the set of polynomial coefficients of $\Bar{f}(\mathscr{D},w)$ and $\Bar{f}(\mathscr{D}',w)$, and $\mathscr{E}=\begin{pmatrix}
    (\frac{1}{2}-y+|z-\Bar{z}|)x_{(1)} \\ \cdots \\ (\frac{1}{2}-y+|z-\Bar{z}|)x_{(d)} \\ \frac{1}{8}x_{(1)}x_{(1)}\\\cdots\\\frac{1}{8}x_{(d)}x_{(d)}
    \end{pmatrix}_{(d+d^2)\times 1}$, where $x_{(e)}$ represents $e$-th element in feature vector $\x$.
    
Then, we have 
\begin{align*}
    \Delta_2 &=\|\mathscr{A}'_1-\mathscr{A}'_2\|_2\\
    &= \|\{\sum_{i=1}^n {\Bar\lambda}_{\phi t_{i}}-\sum_{i=1}^n {\Bar\lambda}_{\phi t_{i}'}\}_{\phi \in \cup_{j=1}^2 \Phi_{j}}\|_2\\
    &=\|\{ {\Bar\lambda}_{\phi t_{n}}- {\Bar\lambda}_{\phi t_{n}'}\}_{\phi \in \cup_{j=1}^2 \Phi_{j}}\|_2\\
    &\leq 2 \max_{t=(\x,z,y)} \|\mathscr{E}\|_2\\
    &= 2 \max_{t=(\x,z,y)}\\ &~~~~~\sqrt{\sum_{j=1}^d((\frac{1}{2}-y+ |z-\Bar{z}|)x_j)^2+ \sum_{1\leq e,l \leq d} (\frac{1}{8}x_{(e)}x_{(l)})^2}\\
    &= \sqrt{\frac{d^2}{16}+9d},
\end{align*}
where $t$ is an arbitrary tuple.
\end{proof}

\begin{theorem}
The output model parameter $\hat{w}$ in ADFC (Algorithm 3) guarantees $(\epsilon,\delta)$-differential privacy, where $\epsilon = \frac{1}{d}\epsilon_s+ \frac{d-1}{d}\epsilon_n$ and $\delta = 1-(1-\delta_s)(1-\delta_n)$.
\end{theorem}
\begin{proof}
Assume that the neighboring datasets $\mathscr{D}$ and $\mathscr{D}'$ differ in the last tuple $t_n$ and $t'_n$.
\begin{align*}
    \Pr\left(\hat{f}(\mathscr{D},w)\right) &=\prod_{\phi\in\Phi_{s} }
    \exp{\left(-\frac{1
    }{2\sigma_s^2}\left(\sum_{i=1}^n \Bar{\lambda}_{\phi t_{i}}
    -\hat{\lambda}_{\phi}\right)^2 \right)}\\
    &~~~~~~~\prod_{\phi\in\Phi_{n} }
    \exp{\left(-\frac{1
    }{2\sigma_n^2}\left(\sum_{i=1}^n \Bar{\lambda}_{\phi t_{i}}
    -\hat{\lambda}_{\phi}\right)^2 \right)}
\end{align*}
\begin{align*}
   &\left|\log \frac{\Pr\left(\hat{f}(\mathscr{D},w)\right)}{\Pr \left(\hat{f}(\mathscr{D}',w)\right)} \right|\\
   &=\left| \frac{1}{2\sigma_s^2}\sum_{\phi\in\Phi_{s}}\left(\left(\sum_{i=1}^n \Bar{\lambda}_{\phi t_{i}}
    -\hat{\lambda}_{\phi}\right)^2-\left(\sum_{i=1}^n \Bar{\lambda}_{\phi t_{i}'}
    -\hat{\lambda}_{\phi}\right)^2\right)\right.\\
    &~~~\left.+ \frac{1}{2\sigma_n^2}\sum_{\phi\in\Phi_{n}}\left(\left(\sum_{i=1}^n \Bar{\lambda}_{\phi t_{i}}
    -\hat{\lambda}_{\phi}\right)^2-\left(\sum_{i=1}^n \Bar{\lambda}_{\phi t_{i}'}
    -\hat{\lambda}_{\phi}\right)^2\right)\right|\\
    &\leq\left| \frac{1}{2\sigma_s^2}\sum_{\phi\in\Phi_{s}}\left(\left(\sum_{i=1}^n \Bar{\lambda}_{\phi t_{i}}
    -\hat{\lambda}_{\phi}\right)^2-\left(\sum_{i=1}^n \Bar{\lambda}_{\phi t_{i}'}
    -\hat{\lambda}_{\phi}\right)^2\right)\right|\\
    &~+\left| \frac{1}{2\sigma_n^2}\sum_{\phi\in\Phi_{n}}\left(\left(\sum_{i=1}^n \Bar{\lambda}_{\phi t_{i}}
    -\hat{\lambda}_{\phi}\right)^2-\left(\sum_{i=1}^n \Bar{\lambda}_{\phi t_{i}'}
    -\hat{\lambda}_{\phi}\right)^2\right)\right|\\
    &= \frac{1}{2\sigma_s^2}\left|\|\mathscr{A}'\|_2^2-\|\mathscr{A}'+\mathscr{B}'\|_2^2\right|\\
    &~~~~~~~~~~+\frac{1}{2\sigma_n^2}\left|\|\mathscr{A}''\|_2^2-\|\mathscr{A}''+\mathscr{B}''\|_2^2\right|
\end{align*}
where we let $\mathscr{A'}=\left\{\sum_{i=1}^n \Bar{\lambda}_{\phi t_{i}}
    -\hat{\lambda}_{\phi} \right\}_{\phi \in \Phi_{s}}$, $\mathscr{B'}=\left\{\sum_{i=1}^n \Bar{\lambda}_{\phi t_{i}'}-\sum_{i=1}^n \Bar{\lambda}_{\phi t_{i}}\right\}_{\phi \in \Phi_{s}} $, $\mathscr{A''}=\left\{\sum_{i=1}^n \Bar{\lambda}_{\phi t_{i}}
    -\hat{\lambda}_{\phi} \right\}_{\phi \in \Phi_{n}}$ and $\mathscr{B''}=\left\{\sum_{i=1}^n \Bar{\lambda}_{\phi t_{i}'}-\sum_{i=1}^n \Bar{\lambda}_{\phi t_{i}}\right\}_{\phi \in \Phi_{n}} $.
    
    Now we will use the fact that the distribution of a spherically symmetric normal is not dependent of the orthogonal basis where its constituent normals are drawn. Thus, we work in two basis aligned with $\mathscr{B}'$ and $\mathscr{B}''$ separately.
    Fix the basis $\mathscr{C}'_1,\cdots, \mathscr{C}'_{\left|\Phi_{s}\right|}$ of $\mathscr{B}'$ and draw $\mathscr{A}'$ by first drawing signed lengths $\mathscr{V}'_{\phi}\sim \mathcal{N}(0,\sigma_s^2)$ for $\phi \in \Phi_{s}$, then let $\mathscr{A}'_{\phi} = \mathscr{V}'_{\phi}\mathscr{C}'_{\phi}$ and $\mathscr{A}'=\sum_{\phi \in\Phi_{s}}\mathscr{A}'_{\phi}$. Fix the basis $\mathscr{C}''_1,\cdots, \mathscr{C}''_{\left|\Phi_{n}\right|}$ of $\mathscr{B}''$ and draw $\mathscr{A}''$ by first drawing signed lengths $\mathscr{V}''_{\phi}\sim \mathcal{N}(0,\sigma_n^2)$ for $\phi \in \Phi_{n}$, then let $\mathscr{A}''_{\phi} = \mathscr{V}''_{\phi}\mathscr{C}''_{\phi}$ and $\mathscr{A}''=\sum_{\phi \in\Phi_{n}}\mathscr{A}''_{\phi}$. 
    
    Without loss of generality, we assume that $\mathscr{C}'_1$ is parallel to $\mathscr{B}'$ and $\mathscr{C}''_1$ is parallel to $\mathscr{B}''$. Based on the triangle with the base $\mathscr{B}'+\mathscr{A}'_1$ and the edge $\sum_{\phi =2}^{\left| \Phi_{s}\right|}\mathscr{A}'_{\phi}$ orthogonal to $\mathscr{B}'$, we have $\|\mathscr{A}'+\mathscr{B}'\|_2^2-\|\mathscr{A}'\|_2^2 = \|\mathscr{B}'\|_2^2+ 2\mathscr{C}'_1\|\mathscr{B}'\|_2$. Since $\|\mathscr{B}'\|_2\leq \frac{1}{d}\Delta_2'$, we have
    $\frac{1}{2\sigma_s^2}\left|\|\mathscr{A}'\|_2^2-\|\mathscr{A}'+\mathscr{B}'\|_2^2\right|\leq \frac{1}{2d\sigma_s^2}\left|{\Delta'}_2^{2}+2|\mathscr{V}'_1|\Delta'_2\right|$. Similarly, consider that the triangle with the base $\mathscr{B}''+\mathscr{A}''_1$ and the edge $\sum_{\phi =2}^{\left| \Phi_{n}\right|}\mathscr{A}''_{\phi}$ orthogonal to $\mathscr{B}''$, we have $\|\mathscr{A}''+\mathscr{B}''\|_2^2-\|\mathscr{A}''\|_2^2 = \|\mathscr{B}''\|_2^2+ 2\mathscr{C}''_1\|\mathscr{B}''\|_2$. Since $\|\mathscr{B}''\|_2\leq \frac{d-1}{d}\Delta_2$, we have
    $\frac{1}{2\sigma_n^2}\left|\|\mathscr{A}''\|_2^2-\|\mathscr{A}''+\mathscr{B}''\|_2^2\right|\leq \frac{d-1}{2d\sigma_n^2}\left|{\Delta'}_2^2+2|\mathscr{V}''_1|\Delta'_2\right|$. 
    When set $|\mathscr{V}'_1|\leq \frac{1}{2}(2\sigma_s^2\epsilon_s-1)$ and $|\mathscr{V}''_1|\leq \frac{1}{2}(2\sigma_n^2\epsilon_n-1)$, we have $ \frac{1}{2d\sigma_s^2}\left|{\Delta'}_2^2+2|\mathscr{V}'_1|\Delta'_2\right| \leq \frac{1}{d}\epsilon_s$ and $\frac{d-1}{2d\sigma_n^2}\left|{\Delta'}_2^2+2|\mathscr{V}''_1|\Delta'_2\right|\leq \frac{d-1}{d}\epsilon_n$. Thus, we have the privacy loss $\left|\log\left(\Pr\left(\hat{f}(\mathscr{D},w)\right)/\Pr \left(\hat{f}(\mathscr{D}',w)\right)\right)\right|\leq \frac{1}{d}\epsilon_s+ \frac{d-1}{d}\epsilon_n = \epsilon$. 
    
    To ensure the privacy loss is bounded by $\epsilon$ with probability at least $1-\delta$, we require 
    \begin{align*}
        &\Pr\left(|\mathscr{V}'_1|\leq \frac{1}{2}(2\sigma_s^2\epsilon_s-1),|\mathscr{V}''_1| \leq \frac{1}{2}(2\sigma_n^2\epsilon_n-1)\right)\\
        &=\Pr\left(|\mathscr{V}'_1|\leq \frac{1}{2}(2\sigma_s^2\epsilon_s-1)\right)\Pr\left(|\mathscr{V}''_1| \leq \frac{1}{2}(2\sigma_n^2\epsilon_n-1)\right)\\
        &\geq 1-\delta.
    \end{align*}
    Now we will give the upper bound of $\Pr\left(|\mathscr{V}'_1|\leq \frac{1}{2}(2\sigma_s^2\epsilon_s-1)\right)$ by using the tail bound of $\mathscr{V}'_{1}\sim \mathcal{N}(0,\sigma_s^2)$. Hence, we have $\Pr\left(\mathscr{V}'_1 > r' \right)\leq \frac{\sigma_s}{\sqrt{2}r'}\exp(-\frac{r'^2}{2\sigma_s^2})$. By letting $r' =\frac{1}{2}(2\sigma_s^2\epsilon_s-1) $ in the above inequality, we have $\Pr\left(\mathscr{V}'_1 > \frac{1}{2}(2\sigma_s^2\epsilon_s-1) \right)\leq \frac{\sqrt{2}\sigma_s}{2\sigma_s^2\epsilon_s-1}\exp\left(-\frac{1}{2}\left(\frac{2\sigma_s^2\epsilon_s-1}{2\sigma_s}\right)^2\right)$. When $\sigma_s\geq\frac{\sqrt{2}\Delta'_2}{2\epsilon_s}(\sqrt{\log (\sqrt{\frac{2}{\pi}}\frac{1}{\delta}})+\sqrt{\log (\sqrt{\frac{2}{\pi}}\frac{1}{\delta_s})+\epsilon_s})$, $\epsilon_s>0$ and $\delta_s$ is very small, we have $\Pr\left(\mathscr{V}'_1 >\frac{1}{2}(2\sigma_s^2\epsilon_s-1) \right)\leq \delta_s/2$. Thus, we can prove that $\Pr\left(|\mathscr{V}'_1| \leq\frac{1}{2}(2\sigma_s^2\epsilon_s-1) \right) \geq 1-\delta_s$. In the same way, we can prove $\Pr\left(|\mathscr{V}''_1| \leq\frac{1}{2}(2\sigma_n^2\epsilon_n-1) \right) \geq 1-\delta_n$. Therefore, if we let $\delta = 1-(1-\delta_s)(1-\delta_n)$, we have
     \begin{align*}
        \Pr\left(|\mathscr{V}'_1|\leq \frac{1}{2}(2\sigma_s^2\epsilon_s-1),|\mathscr{V}''_1| \leq \frac{1}{2}(2\sigma_n^2\epsilon_n-1)\right)\geq 1-\delta,
    \end{align*}
    which proves that the computation of $\hat{f}(\mathscr{D},w)$ satisfies $(\epsilon,\delta)$-differential privacy. Apparently, the final result $\hat{w}$ also satisfies $(\epsilon,\delta)$-differential privacy.
\end{proof}

\end{document}